\documentclass{article}
\usepackage[utf8]{inputenc} 
\usepackage[T1]{fontenc}
\usepackage{url}              
\usepackage{cite}             

\usepackage[cmex10]{amsmath}  
\usepackage{mleftright}       
\mleftright                   

\usepackage{graphicx}         
\usepackage{booktabs}         
\usepackage{mathdots}

\usepackage[utf8]{inputenc} 
\usepackage[T1]{fontenc}
\usepackage{url}
\usepackage{ifthen}
\usepackage{cite}
\usepackage{algorithm}
\usepackage{algpseudocode}
\usepackage[cmex10]{amsmath} 
\usepackage{amssymb}
\usepackage{xcolor}

\usepackage{subcaption}
\usepackage{comment}
\usepackage{tabularx}
\usepackage{mathtools}

\usepackage{tikz}
\usetikzlibrary {arrows.meta,automata,positioning} 

\usepackage{amsfonts, amsthm}

\usepackage{graphicx}%
\usepackage{multirow}%
\usepackage{amsmath,amssymb,amsfonts}%
\usepackage{amsthm}%
\usepackage{mathrsfs}%
\usepackage[title]{appendix}%
\usepackage{xcolor}%
\usepackage{textcomp}%
\usepackage{manyfoot}%
\usepackage{booktabs}%
\usepackage{algorithm}%
\usepackage{algorithmicx}%
\usepackage{algpseudocode}%
\usepackage{listings}%

\usepackage{hyperref}
\usepackage{mathtools}
\usepackage{enumerate}
\usepackage{subcaption}
\usepackage{tabularx}
\usepackage{makecell}

\usepackage{tikz}
\usetikzlibrary{matrix}
\usepackage{ulem}
\usepackage{caption}
\usepackage{lipsum}
\usepackage{parskip}
\usepackage{resizegather}

\makeatletter
\newcommand{\multiline}[1]{%
  \begin{tabularx}{\dimexpr\linewidth-\ALG@thistlm}[t]{@{}X@{}}
    #1
  \end{tabularx}
}
\makeatother

\newtheorem{theorem}{Theorem}

\newtheorem{lemma}{Lemma}
\newtheorem{proposition}{Proposition}
\newtheorem{corollary}{Corollary}

\theoremstyle{definition}
\newtheorem{definition}{Definition}
\theoremstyle{remark}
\newtheorem{remark}{Remark}
\newtheorem{example}{Example}

\newcommand{\bfv}{\mathbf{v}}
\newcommand{\bfu}{\mathbf{u}}

\title{A new method for erasure decoding of convolutional codes}
\author{Julia Lieb, Raquel Pinto, Carlos Vela}
\date{}

\begin{document}

\maketitle
\begin{abstract}
 In this paper, we propose a new erasure decoding algorithm for convolutional codes using the generator matrix. This implies that our decoding method also applies to catastrophic convolutional codes in opposite to the classic approach using the parity-check matrix. We compare the performance of both decoding algorithms. Moreover, we enlarge the family of optimal convolutional codes (complete-MDP) based on the generator matrix.
\end{abstract}

\section{Introduction}\label{sec1}

An erasure channel is a communication channel where, after transmission, each  data symbol is received without error or known to be lost (``erased”). An important example of an erasure channel is the internet and therefore there are many applications over this channel, like video-streaming, that require reliable and fast decoding. In \cite{TRS:Decod} the authors showed that convolutional codes are specially suited for decoding over an erasure channel. Unlike block codes, which process data in fixed-size blocks, convolutional codes offer more flexibility in handling data streams, by selecting properly parts of the received information, called a ``window'', and correcting the erasures in the selected window via a parity-check matrix. Sliding this window along the received information sequence according to the distribution of the erasures it is possible to increase the number of corrected erasures.  
In \cite{TRS:Decod}, this way of erasure decoding was performed using the parity-check matrix of the code. In this paper, we present an erasure decoding algorithm using the generator matrix of a convolutional code. This algorithm can be applied to any convolutional code, also to convolutional codes that do not possess a parity-check matrix.

Maximum Distance Profile (MDP) convolutional codes have optimal recovering rate and were defined in \cite{MDS-Conv} for the class of non-catastrophic convolutional codes, i.e., convolutional codes which possess a parity-check matrix. A particular subclass of MDP convolutional codes are the complete MDP codes, defined in \cite{TRS:Decod}, which behave better than the general class of MDP convolutional codes when large bursts of erasures occur. In \cite{J:CompleteMDP} the author showed the existence of $(n,k,\delta)$ complete MDP convolutional codes for all code parameters $n,k, \delta$ with $(n-k) \mid \delta$ and the nonexistence of these codes if $(n-k) \nmid \delta$. In \cite{AL:j-MDPComplete} the construction of complete MDP convolutional codes over small fields was studied and the more general notion of complete $j$-MDP convolutional codes was defined. These codes perform optimally for decoding with time delay $j$, i.e., considering windows of size $n(j+1)$.

Complete and complete $j$-MDP convolutional codes were defined using parity-check matrices. In this paper, we generalize the notion of complete $j$-MDP convolutional codes to the more general class of delay-free convolutional codes proposing a new definition of $(n,k,\delta)$ complete $j$-MDP convolutional codes when $k \mid \delta$ using the generator matrices of the codes.

The paper is organized as follows. In Section \ref{preliminaries} we give preliminaries on convolutional codes, in particular on MDP and complete MDP codes. In Section \ref{sec_decoding} we show how to decode over the erasure channel using the generator matrix. In Section \ref{sec_CMDP}, for $k\mid\delta$, we define a novel notion of $(n,k,\delta)$ complete MDP and complete $j$-MDP convolutional codes using the generator matrix.
We discuss the decoding capabilities of these codes and we show that if $k \mid \delta$ and $(n-k) \mid \delta$, then the notion of complete $j$-MDP convolutional code defined in this paper coincides with the one defined in \cite{TRS:Decod,AL:j-MDPComplete}. Section \ref{sec_complexity} compares the performance of the decoding algorithm presented in this paper using the generator matrix and the decoding algorithm that makes use of the parity-check matrix introduced in \cite{TRS:Decod}.

\section{Preliminaries}\label{preliminaries}

Let $\mathbb F$ be a finite field and let $\mathbb F[z]$ denote the ring of polynomials with coefficients in $\mathbb F$. Let $\mathbb F[z]^s$ and $\mathbb F[z]^{s \times \ell}$ denote the $s$-dimensional row-vectors and the $s \times \ell$ matrices over $\mathbb F[z]$, respectively. We will also represent the elements of $\mathbb F[z]^s$ as elements of $\mathbb F^s[z]$ and the elements of $\mathbb F[z]^{s \times \ell}$ as elements of $\mathbb F^{s \times \ell}[z]$ depending on the problem we are dealing with.
In this section we will present the definitions and results on convolutional codes that will be needed throughout the paper. For more details see \cite{LPR-ConvCod}.

\begin{definition}
    An $(n,k)$ convolutional code $\cal C$ is an $\mathbb F[z]$-submodule of $\mathbb{F}[z]^n$ of rank $k$. A matrix $G(z) \in \mathbb F[z]^{k \times n}$ whose rows constitute a basis of $\cal C$ is called a generator matrix of $\cal C$, i.e.,
    $$
{\cal C}=
\{v(z) \in \mathbb{F}[z]^n \, : \, v(z)=u(z)G(z) \mbox{ for some } u(z) \in \mathbb{F}[z]^k \}.
    $$
   The vector $v(z)=u(z)G(z)$ is the codeword corresponding to the information sequence $u(z)$.
\end{definition}

A convolutional code $\cal C$ admits many generator matrices. If $G(z),\tilde G(z) \in \mathbb F[z]^{k \times n}$ are two generator matrices of $\cal C$, then
$$
\tilde G(z)=U(z)G(z),
$$
where $U(z) \in \mathbb F[z]^{k \times k}$ is a unimodular matrix. A matrix  $U(z) \in \mathbb F[z]^{k \times k}$ is unimodular if it has an inverse over $\mathbb F[z]$. Moreover, $U(z)$ is unimodular if and only if its determinant belongs to $\mathbb F\backslash \{0\}$ \cite{LPR-ConvCod}.

A matrix $G(z) \in \mathbb F[z]^{k \times n}$ is left prime if
\begin{equation} \label{equivencoders}
G(z)=X(z) \bar G(z),
\end{equation}
for some $X(z) \in \mathbb F[z]^{k \times k}$ and $\bar G(z) \in \mathbb F[z]^{k \times n}$ implies that $X(z)$ is unimodular. Since two generator matrices of a convolutional code $\cal C$ differ by left multiplication with a unimodular matrix, we have that if $\cal C$ admits a left prime generator matrix, then all its generator matrices are left prime and $\cal C$ is said to be non-catastrophic (otherwise, $\cal C$ is called catastrophic). Moreover, the full size minors of any two generator matrices of $\cal C$ differ by left multiplication with a nonzero constant and therefore the maximum degree of their full size minors is the same and it is called the degree of $\cal C$. An $(n,k)$ convolutional code of degree $\delta$ is said to be an $(n,k, \delta)$ convolutional code.

For $i=1,\hdots,k$, the $i$-th row degree of $G(z)\in\mathbb F[z]^{k\times n}$ is defined as the maximum degree of the entries of row $i$ of $G(z)$. Any $(n,k,\delta)$ convolutional code $\cal C$ admits a row reduced generator matrix, i.e. a generator matrix such that the sum of its row degrees equals $\delta$. Any two reduced generator matrices of $\cal C$ have the same row degrees, up to a permutation. Moreover, a matrix $G(z) \in \mathbb F[z]^{k \times n}$ is row reduced if and only if for all $u(z)=[u_1(z) \, u_2(z) \cdots u_k(z)] \in \mathbb F[z]^{k}$, the degree of $u(z)G(z)$ is equal to $max_{i: u_i(z) \neq 0}\{k_i + deg(u_i(z))\}$, where $k_i$ is the $i$-the row degree of $G(z)$, $i=1,2,\dots,k$ \cite{complementaryMinors}.

A full row rank matrix $H(z) \in \mathbb F[z]^{(n-k) \times n}$ is said to be a parity-check matrix of a convolutional code $\cal C$ if for all $v(z) \in \mathbb{F}[z]^n$
$$
v(z) \in {\cal C} \Leftrightarrow H(z)v(z)^T=0.
$$
We write $H(z)=\sum_{i=0}^{\nu}H_iz^i$ with $H_{\nu}\neq 0$ and define $\deg(H(z))=\nu$.
Not all convolutional codes admit a parity-check matrix as the following theorem shows.

\begin{theorem}\cite[Theorem 1.2.8]{LPR-ConvCod}
    An $(n,k)$ convolutional code $\cal C$ admits a parity-check matrix if and only $\cal C$ is non-catastrophic. 
\end{theorem}

The support of a vector $w(z)=\displaystyle \sum_{i \in \mathbb N_0} w_i z^i \in \mathbb F[z]^{\ell}$ is defined as 
$$
{\rm supp}(w(z))=\{i \in \mathbb N_0\, : \, w_i \neq 0\}.
$$
A generator matrix $G(z) \in \mathbb F[z]^{k \times n}$ is said to be delay-free if for every $u(z) \in \mathbb{F}[z]^k$ and every $t \in \mathbb N_0$
$$
{\rm supp}(u(z)G(z)) \subset [t, + \infty [ \Rightarrow {\rm supp}(u(z)) \subset [t, + \infty [.
$$
This means that the supports of an information sequence and the corresponding codeword have the same minimum point, i.e., they ``start" at the same time.
A generator matrix $G(z)$ is delay-free if and only if $G(0)$ is full row rank. 
Note that if $G(0)$ is not full rank, then there exists $u_0\in\mathbb F_q^k$, $u_0\neq 0$, such that $v_0=u_0G(0)=0$. Hence, different $u_0$ can lead to the same $v_0$, which means that, even without erasures, $u_0$ cannot be recovered from $v_0$. Since for two generator matrices
$G(z),\tilde G(z) \in \mathbb F[z]^{k \times n}$ of the same code $\cal C$, there exists a  unimodular matrix $U(z) \in \mathbb F[z]^{k \times k}$ such that $\tilde G(z)=U(z)G(z)$,
we have that $\tilde G(0)=U(0)G(0)$. Therefore, $\tilde G(0)$ is full row rank if and only if $G(0)$ is full row rank since $U(0)$ is invertible.
 Thus, delay-freeness is an invariant of the code and we say that a convolutional code is delay-free if all its generator matrices are delay-free. Note that if $G(z)\in\mathbb{F}[z]^{k\times n}$ is left prime, then $G(0)$ is full row rank, which means that non-catastrophic convolutional codes are necessarily delay-free. 


The capability of error correction of a code is measured by its distance. For convolutional codes two important types of distances are defined: the free distance and the column distances. The weight of a vector $v(z)=\displaystyle \sum_{i \in \mathbb N_0} v_i z^i \in \mathbb{F}[z]^n$ is defined as
$$
{\rm wt}(v(z))= \sum_{i \in \mathbb N_0} {\rm wt}(v_i)
$$
where ${\rm wt}(v_i)$ is the number of nonzero entries of $v_i$. 

\begin{definition}
    Let $\cal C$ be an $(n,k)$ convolutional code. The free distance of $\cal C$ is
    $$
    d_{free}({\cal C})=min\{{\rm wt}(v(z)) \, : \, v(z) \in {\cal C}\backslash \{0\}\}.
    $$
\end{definition}

The next theorem establishes an upper bound on the free distance of a convolutional code called generalized Singleton bound.

\begin{theorem}\cite[Theorem 2.2]{Roxana99}
    Let $\cal C$ be an $(n,k,\delta)$ convolutional code. Then
    $$
    d_{free}({\cal C}) \leq (n-k)\left(\left\lfloor \frac{\delta}{k} \right\rfloor+1\right)   + \delta+1  
    $$
\end{theorem}

When we consider transmission over an erasure channel, column distances have an important role. In \cite{MDS-Conv} the authors restricted the definition of column distances to non-catastrophic convolutional codes but these notions can be easily generalized for the more general class of delay-free convolutional codes. 

Given a vector $v(z)=\displaystyle \sum_{i \in \mathbb N_0} v_i z^i \in \mathbb{F}[z]^n$ and $j_1,j_2 \in \mathbb N_0$ with $j_1 <j_2$, we define the truncation of $v(z)$ to the interval $[j_1,j_2]$, as $v_{[j_1,j_2]}(z)=v_{j_1}z^{j_1} + v_{j_1+1}z^{j_1+1} + \cdots + v_{j_2} z^{j_2}$.

\begin{definition}
    Let $\cal C$ be a delay-free $(n,k,\delta)$  convolutional code. The $j$-column distance of $\cal C$ is defined as 
    $$
d^c_j  =  min\{{\rm wt}(v_{[0,j]}(z)) \, : \; v(z) \in {\cal C} \mbox{ with } v_0 \neq 0\}.
    $$
\end{definition}

The next inequalities follow immediately:
$$
d^c_0 \leq d^c_1 \leq \cdots \leq \lim_{j \rightarrow\infty} d^c_j\leq d_{free}({\cal C})
$$

The following proposition describes how to calculate column distances with the help of the parity-check matrix of a non-catastrophic convolutional code. Given a parity-check matrix $H(z)=\displaystyle \sum_{i \in \mathbb N_0} H_i z^i \in \mathbb F[z]^{(n-k) \times n}$ we define the sliding parity-check matrix
\begin{equation}\label{Hj}
H_j^c=\left[\begin{array}{cccc}
     H_0 &  &  & \\
      H_1 & H_0 &  &  \\
      \vdots & \vdots & \ddots &  \\
     H_j & H_{j-1} & \cdots & H_0
\end{array}
\right],
\end{equation}
for $j \in \mathbb N_0$ where $H_i=0$ for $i>\nu$.

\begin{proposition}\cite[Proposition 2.1]{MDS-Conv}\label{prop1}
Let $j \in \mathbb N_0$, $\cal C$ be a non-catastrophic convolutional code,  $H(z)=\displaystyle \sum_{i \in \mathbb N_0} H_i z^i \in \mathbb F[z]^{(n-k) \times n}$ a parity-check matrix of $\cal C$ and $H_j^c$ the corresponding sliding parity-check matrix. For $d\in\mathbb{N}$, the following statements are equivalent:
\begin{enumerate}
    \item $d_j^c=d$.
    \item None of the first $n$ columns of $H_j^c$ is contained in the span of any other $d-2$ columns of this matrix and one of the first $n$ columns
of $H_j^c$ is in the span of some other $d-1$ columns of that
matrix.
\end{enumerate}
\end{proposition}

Next we show an upper bound for the column distances of a convolutional code. Although this result was proved in \cite[Proposition 2.2, Corollary 2.3]{MDS-Conv} for non-catastrophic convolutional codes, it could be extended to delay-free convolutional codes.
For a generator matrix $G(z)=\displaystyle \sum_{i \in \mathbb N_0} G_i z^i=\sum_{i=0}^{\mu}G_iz^i\in \mathbb F[z]^{k \times n}$ with $G_{\mu}\neq 0$ and $G_i=0$ for $i>\mu$ of a delay-free $(n,k,\delta)$ convolutional code and for $j \in \mathbb N_0$, define the sliding generator matrix as
\begin{equation}\label{Gj}
G_j^c=\left[\begin{array}{cccc}
     G_0 & G_1 & \cdots & G_j \\
      & G_0 & \cdots & G_{j-1} \\
      & & \ddots & \vdots \\
      & & & G_0
\end{array}
\right].
\end{equation}

\begin{theorem}\label{Th:SingletonColumn}
Let $\cal C$ be a delay-free $(n,k,\delta)$ convolutional code. For $j \in \mathbb N_0$,
$$
d^c_j \leq (n-k)(j+1)+1.
$$   
Moreover, if $d^c_j = (n-k)(j+1)+1$ for some $j \in \mathbb N_0$, then $d^c_i = (n-k)(i+1)+1$, for $i \leq j$. 
\end{theorem}

\begin{proof}
We show the first statement of the theorem by induction with respect to $j\in\mathbb N_0$.
    Let $G(z)$ be a generator matrix of $\mathcal{C}$. Since $\mathcal{C}$ is delay-free, $G(0)=G_0$ is full-row rank. Let us assume without loss of generality that the minor constituted by the first $k$ columns of $G_0$ is different from zero and therefore there exists $u_0\in\mathbb{F}^k\setminus\{0\}$ such that $u_0G_0=[0,\dots,0,\ast\,\ast,\dots,\ast]$ has the first $k-1$ entries equal to zero, i.e., ${\rm wt}(u_0G_0)\leq n-k+1$ and therefore $d_0^c\leq n-k+1$ (note that this is a proof for the Singleton bound for $(n,k)$ block codes). Let us assume that $d_j^c\leq(n-k)(j+1)+1$ and let $[u_0,u_1,\dots,u_j]$, $u_i\in\mathbb{F}^k$ with $u_0\not=0$, such that ${\rm wt}([u_0,u_1,\dots,u_j]G_j^c)\leq (n-k)(j+1)+1$. Let 
    $$
    [u_0,u_1,\dots,u_j]\begin{bmatrix}
        G_{j+1}\\G_{j}\\\vdots\\G_1
    \end{bmatrix}=[w_1\,w_2]
    $$
    with $w_1\in\mathbb{F}^k$ and $w_2\in\mathbb{F}^{n-k}$. There exists $u_{j+1}$ such that $u_{j+1}G_0=[-w_1\,\bar{w_2}]$ for some $\bar{w_2}\in\mathbb{F}^{n-k}$, and therefore ${\rm wt}([u_0,u_1,\dots,u_{j+1}]G_{j+1}^c)\leq (n-k)(j+2)+1$. 

    Moreover, if $d^c_j < (n-k)(j+1)+1$, then there exists $[u_0,u_1,\dots,u_j]$, $u_i\in\mathbb{F}^k$ with $u_0\not=0$, such that ${\rm wt}([u_0,u_1,\dots,u_j]G_j^c)< (n-k)(j+1)+1$. In the same way as above, we find $u_{j+1}\in\mathbb F^k$ such that $([u_0,u_1,\dots,u_{j+1}]G_{j+1}^c)< (n-k)(j+2)+1$. This shows the second claim of the theorem.
\end{proof}

Since none of the column distances can be larger than the free distance of $\cal C$, there exists an integer
$$
L=\left\lfloor \frac{\delta}{k} \right\rfloor + \left\lfloor \frac{\delta}{n-k}  \right\rfloor
$$
for which $d^c_j$ can be equal to $(n-k)(j+1)+1$ for $j \leq L$ and $d^c_j < (n-k)(j+1)+1$, for $j > L$. A delay-free $(n,k,\delta)$ convolutional code with $d^c_L = (n-k)(L+1)+1$ is called a Maximum Distance Profile (MDP) code. These codes can be characterized in terms of their generator matrices.

\begin{theorem}\cite[Theorem 2.4]{MDS-Conv}\label{Th:SingBoundcolumnGenerator} Let $\cal C$ be a delay-free $(n,k,\delta)$ convolutional code with generator matrix $G(z)$ and let $G_j^c$ be as defined in (\ref{Gj}). Then the following statements are equivalent:
\begin{enumerate}
\item $d^c_j = (n-k)(j+1)+1$,
\item every $(j+1)k \times (j+1)k$ full-size minor of $G_j^c$ formed
from the columns with indices $ 1 \leq t_1 < \cdots < t_{(j+1)k}$, where $t_{sk+1} > sn$, for $s=1,2, \dots, j$, is nonzero.
\end{enumerate}
\end{theorem}

Non-catastrophic MDP convolutional codes can also be characterized in terms of their parity-check matrices. 

\begin{theorem}\cite[Theorem 2.4]{MDS-Conv} \label{MDS_H} Let $\cal C$ be a non-catastrophic $(n,k,\delta)$ convolutional code with parity-check matrix $H(z)$ and let $H_j^c$ be as defined as in (\ref{Hj}). Then the following statements are equivalent:
\begin{enumerate}
\item $d^c_j = (n-k)(j+1)+1$,
\item every $(j+1)(n-k) \times (j+1)(n-k)$ full-size minor of $H_j^c$ formed from the columns with indices $ 1 \leq r_1 < \cdots < r_{(j+1)(n-k)}$, where $r_{s(n-k)} \leq sn$, for $s=1,2, \dots, j$, is nonzero.
\end{enumerate}
\end{theorem}

In \cite{TRS:Decod} the authors propose a decoding method for non-catastrophic convolutional codes over an erasure channel using a parity-check matrix $H(z)=\sum_{i=0}^{\nu}H_i z^i$ of degree $\nu$ of the code. Let us assume that $v(z)=\sum_{i \in \mathbb N_0} v_i z^i$ is the received codeword and that the coefficients $v_i$ are correct for $i<t$ and that $v_t,\dots v_{t+j}$ may have erasures.  
Then,
\begin{equation}\label{eq:SolveH}
    \left(\begin{array}{cccccc}
     H_\nu   & \dots &   H_0 &      &        &\\
            &  H_\nu   & \dots &   H_0 &        &\\
            &       & \ddots &       & \ddots &\\
            &       &        &  H_\nu   & \dots &   H_0\\
    \end{array}\right)\left(\begin{array}{c}
          v_{t-\nu} ^T\\
          v_{t-\nu+1}^T\\
          \vdots\\
          v_{t}^T
          \\
          \vdots
          \\
          v_{t+j}^T
    \end{array}\right)=\mathbf{0}.
\end{equation}
Let $E$ be the set of indices of the vector $$\bar v=\left[\begin{array}{cccccc}
          v_{t-\nu}  & 
          v_{t-\nu+1} & 
          \cdots &
          v_{t}
        & \cdots & 
          v_{t+j}
    \end{array}\right]$$ where there are erasures and let us represent by $\bar v^{e}$ the vector obtained from $\bar v$ by keeping the entries with indices in $E$ and by $\bar v^{r}$ the vector obtained from $\bar v$ by deleting the entries with indices in $E$. Then the matrix equation (\ref{eq:SolveH}) is equivalent to the system of linear equations in the unknowns $\bar v^{e}$,
   \begin{equation}\label{eq_sys}
    {\cal H}^{e} (\bar v^{e})^T = - {\cal H}^{r} (\bar v^{r})^T,
    \end{equation}
    where ${\cal H}^{e}$ and $ {\cal H}^{r} $ are suitable matrices. The linear system (\ref{eq_sys}) has a solution since only erasures can occur during transmission, but the erasures are recoverable if and only if the solution is unique, i.e., if and only if the matrix ${\cal H}^{e}$ is full column rank.

\begin{theorem}\cite{TRS:Decod}\label{th:ParityDecoding}
    Let $\mathcal{C}$ be an $(n,k,\delta)$ non-catastrophic convolutional code with $d^c_{j}$ the $j$-th column distance. If in any sliding window $v_i,v_{i+1}, \cdots, v_{i+j-1}$ of length $(j+1)n$ at most $d^c_{j}-1$ erasures occur, then we can completely recover the transmitted sequence $v(z)=\sum_{i \in \mathbb N_0}v_i z^i$ by using a parity-check matrix of $\cal C$.
\end{theorem}

When the erasure decoding process occurs in its natural direction, i.e., from left to right, it is called forward decoding.
In \cite{TRS:Decod}, the forward recovering rate per window, $R_j$, is defined as $R_j=\frac{d_j^c-1}{(j+1)n}$, in particular, for MDP codes, $R_j=\frac{(j+1)(n-k)}{(j+1)n}$.
 
\begin{corollary}\cite{TRS:Decod}
  If for a non-catastrophic $(n,k,\delta)$ MDP convolutional code $\mathcal{C}$, in any
  sliding window of length $(j+1)n$ at most $(n-k)(j+1)$
  erasures occur for some $j\in\{0,\dots,L\}$, then full error correction from left to right (i.e. by forward decoding) is possible by using a parity-check matrix of the code. 
\end{corollary}

Even though MDP codes are optimal in the sense of having the maximum possible column distances, for the correction of certain erasure patterns codes with stronger decoding capabilities are needed. Complete MDP convolutional codes, a subfamily of MDP codes, are introduced in \cite{TRS:Decod} and studied in \cite{J:CompleteMDP}. When there are many erasures and it is not possible to correct them all, complete MDP convolutional codes will be able to restart the decoding process after a sufficiently large window of correct symbols, called guard space.

\begin{definition}
    Given $n,k,j$ and $\nu$ integers with $k< n$, the set\\ $(\ell_1,\ell_2,\dots,\ell_{(j+1)(n-k)})$ with $\ell_i\in\{1,\dots,(j+1+\nu)n\}$ and $\ell_i<\ell_s$ for $i<s$ is called a non-trivial set of $(n,k,j,\nu)$-parity-check indices if there exists $H(z)=\sum_{i=0}^\nu H_iz^i\in\mathbb{F}[z]^{(n-k)\times n}$ such that the full-size minor of 
    \begin{equation*} \mathcal{H}_j=\left(\begin{array}{cccccc}
     H_\nu   & \dots &   H_0 &      &        &\\
            &  H_\nu   & \dots &   H_0 &        &\\
            &       & \ddots &       & \ddots &\\
            &       &        &  H_\nu   & \dots &   H_0\\
    \end{array}\right)\in\mathbb F^{(j+1)(n-k)\times (j+1+\nu)n}
    \end{equation*}
   
constituted by the columns with indices $\ell_1,\ell_2,\dots,\ell_{(j+1)(n-k)}$ is different from zero. A full-size minor of $\mathcal{H}_j$ formed by a non-trivial set of $(n,k,j,\nu)$-parity-check indices is then called non-trivial full-size minor of $\mathcal{H}_j$.
\end{definition}

\begin{definition}
A  non-catastrophic $(n,k,\delta)$ convolutional code $\mathcal{C}$ is called complete MDP if $(n-k)\mid\delta$ and $\mathcal{C}$ possesses a parity-check matrix $H(z)$ with $\deg(H(z))=\nu$ such that all non-trivial full-size minors of $\mathcal{H}_L$ are nonzero.

\end{definition}

It is easy to see that to obtain a complete-MDP convolutional code, the parity check matrix has to fulfill very strong conditions, which requires large finite fields \cite{TRS:Decod,J:CompleteMDP}.
Therefore, in \cite{AL:j-MDPComplete} a generalization of complete MDP convolutional codes called complete $j$-MDP convolutional codes was introduced.

\begin{definition}\cite{AL:j-MDPComplete}\label{Def:CompletejMDP}
    Let $(n-k)\mid\delta$ and $\mathcal{C}$ be an $(n,k,\delta)$ convolutional code with left prime parity-check matrix $H(z)=H_0+\cdots+H_{\nu} z^{\nu}$ where $\nu=\frac{\delta}{n-k}$. We call $\mathcal{C}$ a complete $j$-MDP convolutional code if all non-trivial full-size minors
    of $\mathcal{H}_j$ are nonzero.
\end{definition}

\begin{remark}
Non-trivial full-size minors can be defined for the matrix $H_j^c$ (and $G_j^c$) in a similar way.
Moreover, Theorem \ref{MDS_H} states that a convolutional code with parity-check matrix $H(z)$ achieves the upper bound for the $j$-th column distance if and only if all non-trivial full-size minors of $H_j^c$ are non-zero (similar for Theorem \ref{Th:SingBoundcolumnGenerator} and $G_j^c$).
\end{remark}

\begin{lemma}\cite{AL:j-MDPComplete}\label{lemma:nottriviallyzeroH}
    A set of indices $\ell_1,\dots,\ell_{(j+1)(n-k)}$ is a non-trivial set of $(n,k,j,\nu)$-parity-check indices if and only if
    \begin{enumerate}
        \item[(i)] $\ell_{(n-k)s+1}>sn$
        \item[(ii)] $\ell_{(n-k)s}\leq n(s+\nu)$
    \end{enumerate}
    for $s=1,\dots,j$.
\end{lemma}

Note that a complete $L$-MDP convolutional code is the same as a complete MDP convolutional code. Furthermore, it is easy to see that a complete $j$-MDP code reaches the bound of Theorem \ref{Th:SingletonColumn} for all $i\leq j$, that is $d_i^c=(i+1)(n-k)+1$ for $i=0,\dots,j$. 
Up to a time delay of $j$ time instants in the decoding process, these codes have the same erasure correcting capabilities as complete MDP convolutional codes.
Hence, one can say that these codes are optimal for decoding with maximum time delay $j$ and it is easier to accomplish a construction for complete $j$-MDP convolutional codes since the conditions on the parity-check matrix are weaker.

In \cite{AL:j-MDPComplete} it is shown that by using the parity-check matrix of a complete $j$-MDP code it is possible to compute a new guard space when the decoding process stops. In \cite{TRS:Decod} this is also proved for the case $j=L$.

\begin{theorem}\cite[Theorem~6]{AL:j-MDPComplete}\label{Th:CompleteMRD_guardSpace}
    If $\mathcal{C}$ is a complete $j$-MDP convolutional code, then it has the following decoding properties:
    \begin{enumerate}[(i)]
        \item If in any sliding window of length $(j+1)n$ at most $(j+1)(n-k)$ erasures occur, complete recovery is possible.
        \item  If in a window of size $(\nu+j+1)n$ there are not more than $(j+1)(n-k)$ erasures, and if they are distributed in such a way that between positions $1$ and $sn$ and between positions $(\nu+j+1)n$ and $(\nu+j+1)n-sn+1$ for $s=1,2,\dots,j+1$, there are not more than $s(n-k)$ erasures, then full correction of all symbols in this interval is possible by using the parity-check matrix.
    \end{enumerate}
\end{theorem}

 When calculating a new guard space, we obtain a guard space recovering rate $\hat{R}_{j+\nu}$ equal to $\frac{(n-k)(j+1)}{(j+1+\nu)n}$ for $j\in\{0,\dots,L\}$, which is smaller than the corresponding forward recovering rate.

\section{Erasure decoding with the generator matrix} \label{sec_decoding}

In this section we present an alternative decoding method for convolutional codes when communicating over an erasure channel. Following a parallel path as done in \cite{TRS:Decod} where the parity-check matrix is used to this end, we propose to use the generator matrix of the code.

Let $\mathcal{C}$ be an $(n,k,\delta)$ convolutional
code with delay-free generator matrix $G(z)=\sum_{i=0}^{\mu}G_iz^{i}$.
To send a message $u(z)=\sum_{i=0}^{\ell}u_iz^i\in\mathbb F[z]^{k}$ with $\deg(u)=\ell$, we encode it as $u(z)G(z)=v(z)=\sum_{i=0}^{\mu+\ell}v_iz^i$. Note that it is also possible to represent this process block-wise, i.e.
 \begin{align}\label{encoding}
 (u_0 \cdots u_{\ell})
    \left(\begin{array}{cccccc}
        G_0 & \dots & G_\mu  &      &        &\\
            & G_0   & \cdots & G_\mu &        &\\
            &       & \ddots &       & \ddots &\\
            &       &        & G_0   & \cdots & G_\mu\\
    \end{array}\right)=(v_0 \cdots v_{\ell+\mu}).
    \end{align}

In the following, we will consider the decoding of such a message with delay $j$, i.e., latest when $v_{i+j}$ arrives, we want to recover $u_i$.

Next, we present an analogue to Proposition \ref{prop1} showing how the column distances of a convolutional code are related with its generator matrix.

\begin{theorem}\label{Porp:FirstKNotSpan}
Let $\mathcal{C}$ be a convolutional code, $G(z)=\sum_{i\in\mathbb{N}_0}G_iz^i$ a generator matrix of $\mathcal{C}$, $G_j^c$ as defined in (\ref{Gj}) and $d\in\mathbb{N}$, then the following statements are equivalent:
\begin{enumerate}
    \item $d_j^c \geq d$.
    \item None of the first $k$ rows of $(G_j^{c})^r$ is contained in the span of any other set of rows of the same matrix, where $(G_j^{c})^r$ is obtained from the matrix $G^c_j$ by erasing a set of $e$ columns with $e\leq d-1$. 
\end{enumerate}
\end{theorem}

 \begin{proof}
 Let us denote by $g_i$ the rows of $G^c_j$ and by $g_i^{r}$ the rows of the matrix $(G_j^{c})^r$, obtained from $G^c_j$ by erasing a set of $e$ columns with $e \leq d^c_j-1$. Let us assume that there exists $s \in \{1,2,\dots,k\}$ such that $g^{r}_s$ is contained in the span of the other rows of $(G_j^{c})^r$. Then   
  $$
 \sum_{\substack{i=0}}^{(j+1)k}\alpha_ig^{r}_i={\bf 0},
    $$
   for some $\alpha_i \in \mathbb F$, with $\alpha_s \neq 0$, and where ${\bf 0}$ is the all-zero vector. Now, consider the same linear combination of the corresponding rows of the matrix $G_j^c$, i.e.,
     $$
\sum_{\substack{i=0}}^{(j+1)k}\alpha_ig_i=\mathbf{w}.
    $$
    We know that $wt(\mathbf{w})\geq d_j^c$ and therefore at least $d_j^c$ elements of $\mathbf{w}$ are different from zero. Let $\mathbf{w}^{r}$ be the vector obtained from $\mathbf{w}$ by removing the corresponding elements as in $(G_j^{c})^r$. It is easy to see that $wt(\mathbf{w}^{r})>0$ but we also can see that $\mathbf{w}^{r}=\sum_{\substack{i=0}}^{(j+1)k}\alpha_ig_i^{r}=\mathbf{0}$. Therefore, we have a contradiction and the statement holds. 

Conversely, let us assume that $d^c_j < d$ and let $\mathbf{w}=\left[\begin{array}{cccc} u_0 & u_1 & \cdots & u_j \end{array}\right] G^c_j$, with $u_0 \neq 0$, and $wt(\mathbf{w}) = d_j^c$. Let $\mathbf{w}^{r}$ be the vector obtained from $\mathbf{w}$ by erasing the $e=d^c_j \leq d-1$ nonzero components of $\mathbf{w}$ and let $(G_j^{c})^r$ be the matrix obtained from $G^c_j$ by deleting the corresponding columns. Then we can write $\mathbf{w}^{r}=\left[\begin{array}{cccc} u_0 & u_1 & \cdots & u_j \end{array}\right] (G_j^{c})^r=\mathbf{0}$, and since $u_0 \neq 0$, it follows that one of the first $k$ rows of $(G_j^{c})^r$ is contained in the span of the other rows of $(G_j^{c})^r$.
\end{proof}

\begin{proposition}\label{Prop:FullRowRank}
    Let $\mathcal{C}$ be a delay-free $(n,k,\delta)$ convolutional code with generator matrix $G(z)$. Assume that for some $j\in\mathbb N_0$, $d_j^c=(n-k)(j+1)+1$.
   Let $(G_j^{c})^r$ be obtained from the matrix $G_j^{c}$ by erasing $\sum_{t=0}^je_t\leq d_j^c-1$ columns holding that
    $$
    \sum_{t=0}^ie_t\geq d_i^c
    $$
    for all $i\in\{0,\dots,j-1\}$, where $e_t$ is the amount of columns erased in the $t$-th set of $n$ columns of $G_j^{c}$.  
    Then, $(G_j^{c})^r$ is full row rank.
\end{proposition}

\begin{proof}
    We are going to prove this result by contradiction. Assume that the matrix $(G_j^{c})^r$ is not full row rank. Then, there exists a message $u=(u_0,\dots,u_j)\not=\mathbf{0}$ such that $u(G_j^{c})^r=\mathbf{0}$. For this to be possible we need that $wt(uG_j^{c})\leq\sum_{t=0}^je_t$ but, if $u_0\not=\mathbf{0}$, we know that $wt(uG_j^{c})\geq d_j^c$ which is a contradiction by hypothesis since $\sum_{t=0}^je_t\leq d_j^c-1$. 
    
    Now, let us assume that $u_0=\mathbf{0}$, let $h$ be minimal such that $u_h$ is different from the all-zero vector and $\Bar{u}=(u_h,\dots,u_j)$. From the assumptions of the theorem it follows that $\sum_{t=h}^je_t\leq d_{j}^c-d^c_{h-1}-1=(j-h+1)(n-k)-1$, as $d_{h-1}^c=h(n-k)+1$ according to Theorem \ref{Th:SingletonColumn}. Since we assumed that $u(G_j^{c})^r=\mathbf{0}$, we have that $\Bar{u}(G_{j-h}^{c})^r=\mathbf{0}$. Therefore, $wt(\Bar{u}G_{j-h}^{c})\leq \sum_{t=h}^je_t$. On the contrary, we know that $wt(\Bar{u}G_{j-h}^{c})\geq d_{j-h}^c=(j-h+1)(n-k)+1$, which contradicts $\sum_{t=h}^je_t\leq (j-h+1)(n-k)-1$.
\end{proof}

The next theorem is the analogue of Theorem \ref{th:ParityDecoding} when we decode using a generator matrix.

\begin{theorem}\label{th:GeneratorDecoding}
Let $\mathcal{C}$ be a delay-free $(n,k,\delta)$ convolutional code with $d^c_j$ the $j$-th column distance. If for some $j\in\mathbb N_0$, in any sliding window $v_i,v_{i+1}, \dots, v_{i+j}$ of length $(j+1)n$ at most $d^c_{j}-1$ erasures occur, then we can completely recover the transmitted sequence $v(z)=\sum_{i \in \mathbb N_0} v_i z^i$ by using a generator matrix $G(z)$ of $\cal C$.
\end{theorem}

\begin{proof}
Assume we have been able to decode the messages correctly up to a time instant $t-1$, i.e., $u_0,\dots,u_{t-1}$ are known, then we have the following system of linear equations: 
 \begin{gather}\label{eq:GdecodSys}
 (u_{t-\mu},\dots, u_{t-1}\ |\ u_t, \dots, u_{t+j})
    \left(\begin{array}{cccccc}
          G_{\mu}  &  \dots  & G_{\mu+j}  \\
          \vdots   &         & \vdots     \\
          G_1      &  \cdots & G_{j+1}    \\
            \hline
          G_0      &  \cdots & G_{j}      \\
                   & \ddots  & \vdots     \\
                   &         & G_0
    \end{array}\right)=(v_{t}, \dots, v_{t+j})
\end{gather}
    
The vector $(u_t, \dots, u_{t+j})$ consists of the unknowns we want to recover and each component of the vector $v_{t-\mu},\dots, v_{t+j}$ which is received correctly gives us one equation that we can use for the recovery. Let $\bar E$ be the set of indices of the vector $(v_t, \dots, v_{t+j})$ where there are no erasures and $\bar E_i=\bar E \cap [0, (i+1)n]$, $i=0,1, \dots,j$. As $u_{t-\mu},\hdots,u_{t-1}$ are assumed to be known, for the decoding, we have to solve a system of linear equations of the form

 \begin{gather}\label{system}
 (u_t, \dots, u_{t+j})
    \left(\begin{array}{ccc}
                   G_0   & \cdots & G_{j}\\
               & \ddots & \vdots\\
               & & G_0
\end{array}\right)^{r_j}=
(v_t, \dots, v_{t+j})^{r_j}-  (u_{t-\mu},\dots,u_{t-1})\begin{pmatrix} G_{\mu}    &    \dots     & G_{\mu+j}\\
                      \vdots  &  & \vdots\\
            G_1   & \cdots & G_{j+1}\end{pmatrix}^{r_j} 
    \end{gather}

where the superscript $r$ ($r_i$) indicates that we only consider the columns of the matrix with indices in $\bar E$ $(\bar E_i)$. The right-hand side of this equation is known, hence it can be used to recover the unknowns. Let $e_i$ be the number of erased columns in the $i$-th set of $n$ columns of $G^c_{j}$, then by hypothesis we have that $\sum_{i=t}^{t+j}e_i\leq d_{j}^c-1$.

We will see that we can recover $u_t$ such that it is unique. If $e_t\leq d_0^c-1$, then we can uniquely solve the linear system 
$$
u_tG_0^{r_0}=v_t^{r_0}-(u_{t-\mu},\dots,u_{t-1})\begin{pmatrix} G_{\mu} \\ \vdots\\ G_1\end{pmatrix}^{r_0},
$$
since $G_0^{r_0}$ is a full row rank matrix by Theorem \ref{Porp:FirstKNotSpan}. 
Otherwise, if $e_t\geq d_0^c$, let $h\in\{0,\dots,j\}$ be the first integer such that $\sum_{i=0}^{h}e_i\leq d_{h}^c-1$, i.e., the first time instant for which the amount of erasures is correctable. We know that this element exists because $\sum_{i=t}^{t+j}e_i\leq d_{j}^c-1$ by hypothesis. Then, solving the following linear system:
 \begin{gather}\label{Eq:FirstMomentSolvable}
(u_t, \dots, u_{t+h})
    \left(\begin{array}{ccc}
                   G_0   & \cdots & G_{h}\\
               & \ddots & \vdots\\
               & & G_0
    \end{array}\right)^{r_h}=(v_t, \dots, v_{t+h})^{r_h}-  (u_{t-\mu},\dots,u_{t-1})\begin{pmatrix} G_{\mu}    &    \dots     & G_{\mu+h}\\
                      \vdots  &  & \vdots\\
            G_1   & \cdots & G_{h+1}\end{pmatrix}^{r_h} 
    \end{gather}
allows us to recover uniquely $u_t$ thanks to Theorem \ref{Porp:FirstKNotSpan}. After the recovery of $u_t$, one can slide along the received sequence and proceed with the decoding.
Note that according to Proposition \ref{Prop:FullRowRank}, if $d_h^c=(n-k)(h+1)+1$, we can directly recover $(u_t,\dots,u_{t+h})$ in one step.
\end{proof}

In the following, we give a pseudocode for the algorithm of the preceding theorem and illustrate it with an example. To make it easier, we consider a reduced generator matrix.


\begin{algorithm} 
\caption{Decoding algorithm with the generator matrix for convolutional codes} 
\label{Gen_Decoding}
\begin{algorithmic}[1]
 \Require Delay-free reduced generator matrix $G(z)$; received word $v(z)$ (possibly containing erasures)
 \Ensure Message $u(z)=\sum _\ell u_\ell z^\ell$ 
 \State $t=0$
 \State $j=0$
 \State $S=[\,]$ (an empty list)
 \While{$t+j<\deg{v(z)}-\deg{G(z)}$}
    \While{$\left(\sum_{\ell=t}^{t+j}e_\ell\geq d_j^c\right)\wedge\left(d_j^c<d_{free}\right)$}
        \State $j= j+1$
    \EndWhile
    \If{$\sum_{\ell=t}^{t+j}e_\ell\leq d_j^c-1$}
        \State $(u_t,\dots,u_{t+i})=$ the solution of system \eqref{system} for maximum possible $i\in\mathbb N_0$
        \State append $(u_t,\dots,u_{t+i})$ at the end of $S$
        \State $t=t+i+1$
        \State $j=0$
    \ElsIf{$d_j^c \geq d_{free}$}
        \State $t=\deg{v(z)}-\deg(G(z))$
    \EndIf
\EndWhile
\State return $S$
 
\end{algorithmic}
\end{algorithm}

\newpage

\begin{example}
Consider the $(5,2,2)$ convolutional code with reduced generator matrix $G(z)=G_0+G_1z\in\mathbb F_2[z]^{k\times n}$ where $G_0=\begin{pmatrix}
    1 & 1 & 0 & 1 & 1\\
    1 & 0 & 1 & 1 & 0
\end{pmatrix}$ and $G_1=\begin{pmatrix}
    1 & 1 & 1 & 1 & 1\\
    0 & 0 & 0 & 1 & 1
\end{pmatrix}$. One calculates $d^c_0=3$ and $d^c_i=5$ for $i\geq 5$.

Assume that we send the message $u(z)=(1+z^2,1+z^3)$, which corresponds to the codeword $v(z)=(z+z^2,1+z+z^2+z^3,z+1,z^2+z^4,1+z^2+z^3+z^4)$. Assume further the following erasure pattern: $v_0=(0,1,\not{1},\not{0},1)$, $v_1=(\not{1},1,1,0,\not{0})$, $v_2=(1,1,0,\not{1},1)$, $v_3=(0,\not{1},\not{0},0,\not{1})$, $v_4=(0,0,0,1,\not{1})$.
As in any window of size $2n=10$ there are not more than $d^c_1-1=4$ erasures, we can apply the preceding theorem with $j=1$.

The first system of equations we need to solve is 
$$(u_0\ u_1)\cdot \begin{pmatrix}
    1 & 1 & 1 & \vline & 1 & 1 & 1\\
    1 & 0 & 0 & \vline & 0 & 0 & 1\\
    0 & 0 & 0 & \vline & 1 & 0 & 1\\
    0 & 0 & 0 & \vline & 0 & 1 & 1
\end{pmatrix}=\begin{pmatrix}
    0 & 1 & 1 &\vline & 1 & 1 & 0
\end{pmatrix},
$$ which gives us $u_0=( 1\ 1)$ and $u_1=(0\ 0)$. Note that here we can recover $u_0$ and $u_1$ in one step, which is not guaranteed by the distance properties of the code. In the next step, we will see that in general it is only possible to recover one coefficient of the message per step. For the recovery of $u_2$, we consider the equations
\begin{align*}&(u_2\ u_3)\cdot \begin{pmatrix}
    1 & 1 & 0   & 1& \vline & 1 & 1\\
    1 & 0 & 1  & 0& \vline & 0 & 1\\
    0 & 0 & 0  & 0& \vline & 1 & 1\\
    0 & 0 & 0  & 0& \vline & 1 & 1
\end{pmatrix} \\& =  \begin{pmatrix}
    1 & 1 & 0 & 1 &\vline & 0 & 0
\end{pmatrix}-u_1 \cdot \begin{pmatrix}
    1 & 1 & 1   & 1& \vline & 0 & 0\\
    0 & 0 & 0  & 1& \vline & 0 & 0
\end{pmatrix}
 \\ &=  \begin{pmatrix}
    1 & 1 & 0 & 1  &\vline & 0 &0
\end{pmatrix},
\end{align*}
and obtain $u_2=(1\ 0)$. The preceding equations give us only $u_3\cdot \begin{pmatrix}
    1 & 1\\ 1 & 1 
\end{pmatrix}=\begin{pmatrix}
    1 & 1 
\end{pmatrix}$, which is not enough information to recover $u_3$. Hence, to recover $u_3$, we need to go one step further. Since we know $\deg(v(z))=4$ as well as $\deg(G(z))=1$, it is easy to see that $\deg(u(z))=3$ and thus, $u_i=0$ for $i>3$. This gives us the equations
\begin{align*}
&u_3\cdot\begin{pmatrix}
   1 & 1 & \vline & 1 & 1 & 1 & 1\\
   1 & 1 & \vline & 0 & 0 & 0 & 1
\end{pmatrix}\\&=\begin{pmatrix}
  0 & 0 & \vline & 0 & 0 & 0 & 1
\end{pmatrix}-(1\ 0)\cdot \begin{pmatrix}
    1 & 1 & \vline & 0 & 0 & 0 & 0\\
    0 & 1 & \vline & 0 & 0 & 0 & 0
\end{pmatrix}=(1\ 1\ 0\ 0\ 0\ 1)\end{align*}
with the unique solution $u_3=(0\ 1)$.
In this way we recovered the complete message.
\end{example}

\begin{corollary}
  If for an $(n,k,\delta)$ MDP convolutional code $\mathcal{C}$, in any
  sliding window of length $(j+1)n$ at most $(n-k)(j+1)$
  erasures occur for some $j\in\{0,\dots,L\}$, then full error correction from left to right (i.e. by forward decoding) is possible by using the generator matrix.
\end{corollary}

\begin{example}\label{Ex:MDP1}
Consider a $(3,2,18)$ MDP convolutional code over an erasure channel. Its forward recovering rate $R_L$ is $\frac{28}{84}$ since $L=27$. Assume we receive the following pattern of erasures:
\begin{multline*}
  \cdots\bfv\bfv|\overbrace{\star\star\dots\star}^{20}\overbrace{\bfv\bfv\cdots\bfv}^{42}\overbrace{\star\star\cdots\star}^{14}\overbrace{\bfv\bfv\cdots\bfv}^{8}|\\\overbrace{\star\star\dots\star}^{11}\overbrace{\bfv\bfv\cdots\bfv}^{18}\overbrace{\star\star\bfv\bfv\bfv\star\star\bfv\bfv\bfv\cdots\star\star\bfv\bfv\bfv}^{55}|\cdots
\end{multline*}
where $\star$ refers to an erased component of the vectors, $\bfv$ to a correctly received component and $|$ marks a window of size $84$. We assume that all the previous messages have been correctly decoded. It is noticeable that the amount of erasures in this window is not correctable, nevertheless we can consider a shorter window (see Theorem \ref{th:GeneratorDecoding}). In this case, to correct the first set of erasures we can consider the window of length $60$ ($j=19$) together with the window of $27$ ($n\mu$) previous symbols, which are known.
\begin{equation*}
    \cdots\bfv\bfv|\overbrace{\star\star\dots\star}^{20}\overbrace{\bfv\bfv\cdots\bfv}^{40}\cdots
\end{equation*}
After correcting these erasures we can slide the window to try to recover the next set of $14$ lost symbols. It is easy to see that, even though considering an MDP code, performing forward correction in this case is not possible.
\end{example}

\section{Complete MDP convolutional codes from the generator matrix}\label{sec_CMDP}

MDP convolutional codes are optimal for decoding when full recovery from left to right is possible. In this section, we want to consider the scenario of having too many erasures for recovery in some part of the sequence which causes the complete forward decoding to fail due to the lack of a guard space. However, we want to see how and under which conditions it is possible to compute a new guard space to continue with the decoding even if some earlier parts of the received and the information sequence had to be declared as lost. 

From equation (\ref{eq:GdecodSys}) we can see that when decoding we need to find a block of $n\mu$ correct symbols $v_{t-\mu}$ to $v_{t-1}$ ahead of a block of $(j+1)n$ symbols, $v_t,\dots,v_{t+j}$ where no more than $d^c_j\leq(j+1)(n-k)$ erasures occur, in order to continue the recovery process. In other words, we need to have some guard space.

Let $\mathcal{C}$ be an $(n,k,\delta)$ convolutional code with generator matrix $G(z)=\sum_{i=0}^\mu G_iz^i$, $u(z)=\sum_{i=0}^\ell u_iz^i\in\mathbb{F}[z]^k$ and $u(z)G(z)=v(z)=\sum_{i=0}^{\mu+\ell} v_iz^i\in\mathbb{F}[z]^n$ and 
\begin{equation}
    \mathcal{G}^c_j=\begin{pmatrix}
        G_{\mu} & & \\
                \vdots & \ddots & \\
                       G_0 &  & G_\mu \\
        & \ddots   & \vdots  \\
        &  & G_0       
    \end{pmatrix}\in\mathbb{F}^{(j+1+\mu)k\times n(j+1)}
\end{equation}

\begin{definition}
    Given $n,k,j$ and $\mu$ integers with $k< n$, the set\\ $(\ell_1,\ell_2,\dots,\ell_{(j+1+\mu)k})$ with $\ell_i\in\{1,\dots,(j+1)n\}$ and $\ell_i<\ell_s$ for $i<s$, is called a non-trivial set of $(n,k,j,\mu)$-generator indices if there exists $G(z)=\sum_{i=0}^\mu G_iz^i\in\mathbb{F}[z]^{k\times n}$ such that the full-size minor of $\mathcal{G}^c_j$ constituted by the columns with indices $\ell_1,\ell_2,\dots,\ell_{(j+1+\mu)k}$ is different from zero. A full-size minor of $\mathcal{G}^c_j$ formed by a non-trivial set of $(n,k,j,\mu)$-generator indices is then called non-trivial full-size minor of $\mathcal{G}^c_j$.
\end{definition}

For the recovery of $u_i$ with delay $j$ without assuming that $u_t$ with $t<i$ are known (i.e. to compute a guard space), the equations we need to consider have the following form:
\begin{align}\label{smallcomplete}
    ( u_{i-\mu} \cdots u_i \cdots u_{i+j})\left(\mathcal{G}_j^{c}\right)^r=(v_i \cdots v_{i+j})^{r},
\end{align}
where the superscript $r$ indicates that we only consider the columns of the matrix that correspond to the correctly received components. If $\left(\mathcal{G}_j^{c}\right)^r$
has full row rank, it is possible to recover $u_{i-\mu}, \dots, u_{i+j}$. Of course, this is only possible if $(j+1+\mu)k\leq n(j+1)$ and there are not more than $n(j+1)-(j+1+\mu)k=(n-k)(j+1)-\mu k$ erasures in $v_i, \dots, v_{i+j}$.

If the system in (\ref{smallcomplete}) cannot be solved, we can take a larger window and consider the equations:
\begin{equation}\label{form:bigSystem}
    ( u_{i-2\mu} \cdots u_i \cdots u_{i+j})\left(\mathcal{G}_{\mu+j}^{c}\right)^r=(v_{i-\mu} \cdots v_{i+j})^{r}.
\end{equation}

Using these equations, we can correct up to $(n-k)(j+1)+(n-2k)\mu$ erasures in a window of size $n(j+1+\mu)$.
The optimal codes for this kind of decoding  have the property that all non-trivial full-size minors of $\left(\mathcal{G}_{\mu+j}^{c}\right)^r$ are nonzero. Note that this is only possible if $G_{\mu}$ has full row rank, which in turn is only possible if $k\mid\delta$.

In light of the previous considerations, we extend the definition of complete $j$-MDP codes to $(n,k,\delta)$ convolutional codes with $k\mid\delta$ using the generator matrices of the codes.

\begin{definition}\label{Def:Complete-j-MDP-New}
An $(n,k,\delta)$ convolutional code $\mathcal{C}$ is called complete $j$-MDP,
\begin{enumerate}
    \item[(a)] if $(n-k)\mid\delta$ and $\mathcal{C}$ possesses a parity-check matrix $H(z)=\sum_{i=0}^{\nu}H_iz^i$ where $\nu=\frac{\delta}{n-k}$ such that all non-trivial full-size minors of $\mathcal{H}^c_j$ are nonzero, \textbf{or}
    \item[(b)] if $k\mid\delta$ and $\mathcal{C}$ possesses a generator matrix $G(z)=\sum_{i=0}^{\mu}G_iz^i$ where $\mu=\frac{\delta}{k}$ such that all non-trivial full-size minors of $\mathcal{G}^c_{\mu+j}$ are nonzero.
\end{enumerate} 
A complete $L$-MDP convolutional code is also just called complete MDP convolutional code.
\end{definition}



\begin{theorem}\label{Prop:NotTriviallyZeroMinorsG}
   A full-size minor of $ \mathcal{G}^c_{\mu+j}$ formed by the columns $\ell_{1}<\dots<\ell_{(j+1+2\mu)k}$ is non-trivial if and only if
   \begin{enumerate}[(i)]
        \item $\ell_{sk}\leq sn$
        \item $\ell_{(\mu+s)k+1}>sn$
   \end{enumerate}
   for $s\in\left\lbrace1,2,\dots,j+\mu\right\rbrace$.
\end{theorem}

\begin{proof}
 To obtain a non-trivial minor it is necessary to pick at least $sk$ columns out of the first $sn$ columns of $\mathcal{G}^c_{\mu+j}$ for $s\in\left\lbrace1,2,\dots,j+\mu\right\rbrace$ because the last $(\mu+j+1)n-sn$ columns have only zeros in the first $sk$ rows, i.e. condition (i) is necessary. In a similar way we can conclude that we need to pick at least $sk$ out of the last $sn$ columns because the other columns have the last $sk$ rows equal to zero, which shows that also condition (ii) is necessary to obtain a non-trivial minor. 
 
On the other hand, condition (i) ensures that we cannot pick more than $(j+1+2\mu-s)k$ columns from the last $(j+1+\mu-s)n$ columns of $\mathcal{G}^c_{\mu+j}$, which ensures that these chosen $(j+1+2\mu-s)k$ columns can be complemented to a non-trivial full-size minor of $\mathcal{G}^c_{\mu+j}$. Similarly, condition (ii) ensures that we take at most $(\mu+s)k$ columns out of the first $sn$ columns, which ensures that these chosen $(\mu+s)k$ columns can be complemented to a non-trivial full-size minor of $\mathcal{G}^c_{\mu+j}$.
\end{proof} 

With the preceding theorem we obtain that the requirements on the error pattern which 
can be corrected by a complete $j$-MDP convolutional code are given by the following theorem, which is an analogue of Theorem \ref{Th:CompleteMRD_guardSpace} for the case when the generator matrix is used.

\begin{theorem}\label{Th:CompleteMRD_guardSpace_Generator}
    If $\mathcal{C}$ is a complete $j$-MDP convolutional code as in Definition \ref{Def:Complete-j-MDP-New}(b), i.e. in particular $k\mid\delta$, then it has the following decoding properties:
    \begin{enumerate}[(i)]
        \item If in any sliding window of length $(j+1)n$ at most $(j+1)(n-k)$ erasures occur, complete forward recovery is possible.
        \item  If in a window of size $(\mu+j+1)n$ there are not more than $(n-k)(j+1)+(n-2k)\mu$ erasures, and if they are distributed in such a way that between positions $1$ and $sn$ and between positions $(\mu+j+1)n$ and $(\mu+j+1)n-sn+1$
        for $s=1,2,\dots,j+\mu$, there are not more than $s(n-k)$ erasures, then full correction of all symbols in this interval will be possible by using the generator matrix.
    \end{enumerate}
\end{theorem}

From the previous theorem it is easy to see that also for  complete $j$-MDP convolutional codes forward
recovering rate and guard space recovering rate are different. On one hand, when forward recovery is possible as in case \textit{(i)}, the recovering rate is $R_i=\frac{(i+1)(n-k)}{(i+1)n}$ for $i\in\{0,1,\dots,j\}$. On the other hand, when computing a new guard space is necessary and we must invoke property \textit{(ii)} of Theorem \ref{Th:CompleteMRD_guardSpace_Generator}, we have a decreased recovering rate of $\hat{R}_{i+\mu}=\frac{(i+1)(n-k)+(n-2k)\mu}{(i+1+\mu)n}$ for $i\in\{0,1,\dots,j\}$. Table \ref{tab:RcoveringRatio} compares the recovering rates when computing a guard space using the generator matrix or the parity-check matrix. It shows that if $k<n-k$, then we can recover more erasures using the generator matrix and if $n-k<k$, we can recover more with the parity-check matrix and if $k=n-k$, the two recovering rates are the same. Note, however, that for these recovering rates, we need to assume $k\mid\delta$ when using $G(z)$ and $(n-k)\mid\delta$ when using $H(z)$.
\begin{table}[h]
    \centering
    {\renewcommand{\arraystretch}{2}
    \begin{tabular}{ccc}
                 & With $G(z)$ and $k\mid\delta$                         & With $H(z)$ and $(n-k)\mid\delta$\\
    \cmidrule{2-3}
    Guard space recovering rate    & {$\frac{(n-k)(j+1)+(n-2k)\mu}{(j+1+\mu)n}$} &{ $\frac{(n-k)(j+1)}{(j+1+\nu)n}$}\\
    \cmidrule{2-3}
    Best parameter range     & {$k\leq n-k$}                                  & {$k\geq n-k$}\\
    \cmidrule{2-3}
    \end{tabular}
    }
    \caption{Recovering rate comparison}
    \label{tab:RcoveringRatio}
\end{table}

\begin{example}
Consider a $(3,1,18)$ complete MDP convolutional code over an erasure channel with $\mu=\delta/k=18$. Since $L=27$, we have that its forward recovering rate is $R_L=\frac{(L+1)(n-k)}{(L+1)n}=\frac{56}{84}$. Now, we suppose we received the following messages:
    \begin{multline*}
    \cdots\star\star|\overbrace{\star\star\dots\star}^{20}\overbrace{\bfv\bfv\cdots\bfv}^{11}\overbrace{\star\star\cdots\star}^{15}\overbrace{\bfv\bfv\cdots\bfv}^{19}\overbrace{\star\star\dots\star}^{19}|\\
    |\overbrace{\bfv\bfv\bfv}^{3}\overbrace{\star\star\cdots\star}^{16}\overbrace{\bfv\bfv\cdots\bfv}^{8}\overbrace{\star\star\dots\star}^{24}\overbrace{\bfv\bfv\cdots\bfv}^{23}\overbrace{\star\star\dots\star}^{10}|\\
    |\overbrace{\star\star\dots\star}^{17}\overbrace{\bfv\bfv\cdots\bfv}^{23}\overbrace{\star\star\cdots\star}^{21}\overbrace{\bfv\bfv\cdots\bfv}^{4}\overbrace{\star\star\dots\star}^{19}|\dots
    \end{multline*}
    Since we do not have a window of at least $54$ correctly received symbols (i.e. a guard space), we cannot use the forward decoding process with MDP convolutional codes. It is easy to see that there is no window of size $84$ with at most $38$ erasures, so we cannot decode as in formula (\ref{smallcomplete}). When enlarging the considered window the recovering rate is reduced to $\hat{R}_{138}=\frac{74}{138}$ (see top left formula in Table \ref{tab:RcoveringRatio}). Since the $(L+1+\mu)n=138$-size window 
    \begin{multline*}
    \cdots\overbrace{\bfv\bfv\cdots\bfv}^{11}\overbrace{\star\star\dots\star}^{15}\overbrace{\bfv\bfv\cdots\bfv}^{19}\overbrace{\star\star\cdots\star}^{19}|\overbrace{\bfv\bfv\bfv}^{3}\overbrace{\star\star\cdots\star}^{16}\overbrace{\bfv\bfv\cdots\bfv}^{8}\overbrace{\star\star\dots\star}^{24}\overbrace{\bfv\bfv\cdots\bfv}^{23}\dots
    \end{multline*}
    has $74$ erasures (fulfilling condition (ii) of Theorem \ref{Th:CompleteMRD_guardSpace}), we can proceed with the recovering of the information by solving system \eqref{form:bigSystem}.
    As we now have a big enough window of correct information, we can decode the next bursts of length $10$, $17$ and $21$:
    \begin{equation*}
        \overbrace{\bfv\bfv\cdots\bfv}^{54}\overbrace{\star\star\dots\star}^{10}|\overbrace{\star\star\dots\star}^{17}\overbrace{\bfv\bfv\cdots\bfv}^{23}\overbrace{\bfv\bfv\cdots\bfv}^{21}\cdots
    \end{equation*}
    by using the same method as in Example \ref{Ex:MDP1} (forward correction). At the end of the recovering process we were not able to recover the first ($\cdots|\overbrace{\star\star\dots\star}^{20}\cdots$) burst of erasures and the decoding of the last ($\cdots\overbrace{\star\star\dots\star}^{19}|\cdots$) burst of erasures will depend on how the next symbols arrive.
\end{example}

The following lemma is a straight-forward extension of \cite[Lemma~3]{AL:j-MDPComplete} using the extended definition of complete $j$-MDP code as in Definition \ref{Def:Complete-j-MDP-New}.

\begin{lemma}\cite{AL:j-MDPComplete}
    If for some $j\in\{0,1,\dots,L\}$, $\mathcal{C}$ is a complete $j$-MDP convolutional code, then $\mathcal{C}$ is also a complete $i$-MDP convolutional code for all $i\leq j$.
\end{lemma}

Note that this lemma implies that complete $j$-MDP codes are also optimal for decoding using equation \eqref{smallcomplete}. 

Moreover, when not considering the delay in the decoding process, the optimal codes are those where all non-trivial full-size minors of
$$\mathcal{G}^c_{L+\mu}=\begin{pmatrix}
        G_{\mu} & & \\
                \vdots & \ddots & \\
                       G_0 &  & G_\mu \\
        & \ddots   & \vdots  \\
        &  & G_0       
    \end{pmatrix}\in\mathbb F^{k(L+1+2\mu)\times n(L+1+\mu)}$$
are nonzero, i.e. the complete MDP convolutional codes. 

It is important to notice that complete $j$-MDP codes are in general not MDP.
However, in \cite{TRS:Decod} it is proved that every complete MDP code defined as in Definition \ref{Def:CompletejMDP} with $j=L$ is an MDP code. Next, we will show that this is also true for our new complete MDP codes defined via the generator matrix, in case that $k\leq n-k$, which is the parameter range where the codes defined via the generator matrix perform better; see Table \ref{tab:RcoveringRatio}.

\begin{theorem}
    Each complete MDP convolutional code with $k\mid\delta$ and $k\leq n-k$, or $(n-k)\mid\delta$  is MDP.
\end{theorem}

\begin{proof}
    For $(n-k)\mid\delta$, the statement is known; see \cite{TRS:Decod}. So let $k\mid\delta$, $G(z)=\sum_{i=0}^{\mu}G_iz^i\in\mathbb{F}[z]^k$ be a row reduced generator matrix of $\mathcal{C}$. Then $\mu=\frac{\delta}{k}$ and $L\geq\mu$. Assume that all non-trivial full-size minors of 
    \begin{equation}
 \mathcal{G}^c_{\mu+L}=\begin{pmatrix}
        G_{\mu} & &  &\vline & & & &  &  \\
                \vdots & \ddots & &\vline\\
                       G_0 & \cdots & G_\mu & \vline\\
        & \ddots   & \vdots & \vline &  G_\mu \\
        &  & G_0   & \vline  & \vdots   & \ddots \\
        \hline  
        & & & \vline & G_0 & \cdots & G_\mu & \cdots & G_L\\
         & & & \vline & & \ddots & \vdots & \ddots & \vdots\\
          &  & & \vline & & & G_0 & \cdots & G_\mu\\
           & & & \vline & & & & \ddots & \vdots \\
            & & & \vline &  & & & & G_0
    \end{pmatrix}\in\mathbb F^{k(L+1+2\mu)\times n(L+1+\mu)}
\end{equation}
are nonzero. We choose $k(L+1+2\mu)$ out of the $n(L+1+\mu)$ columns of $\mathcal{G}^c_{\mu+L}$ forming a matrix $C$ as follows: first, we choose $2\mu k$ out of the first $n\mu$ columns building a matrix $A$ such that it corresponds to a non-trivial full-size minor of 
$$
    \begin{pmatrix}G_{\mu} & & 0 \\
                \vdots & \ddots & \\
                       G_0 & \cdots & G_\mu \\
        & \ddots   & \vdots \\
     0   &  & G_0   
     \end{pmatrix}\in\mathbb F^{2\mu k\times n\mu}.
$$ 
Second, we pick $k(L+1)$ out of the last $n(L+1)$ columns of $\mathcal{G}^c_{\mu+L}$ forming a matrix $B$ such that it corresponds to a non-trivial full-size minor of 
$$
\mathcal{G}^c_L=\begin{pmatrix} 
         G_0 & \cdots & G_\mu  & \cdots & G_L    \\
             & \ddots & \vdots & \ddots & \vdots \\
             &        & G_0    & \cdots & G_\mu  \\
             &        &        & \ddots & \vdots \\
             &        &        &        & G_0
    \end{pmatrix}.
$$
Now, we have that $\det(C)=\det(A)\cdot\det(B)$ and since $\det(C)$ and $\det(A)$ are nonzero, also $\det(B)$ is nonzero. Consequently, all non-trivial full-size minors of $\mathcal{G}^c_L$ are nonzero, i.e. the corresponding code is MDP.
\end{proof}

\begin{corollary}
    Each complete MDP convolutional code with $k\mid\delta$ and $k\leq n-k$, or $(n-k)\mid\delta$ is non-catastrophic.
\end{corollary}

\begin{proof}

For $(n-k)\mid\delta$, the statement is known; see \cite{TRS:Decod}. Moreover, according to the previous theorem, each complete MDP convolutional code with $k\mid\delta$ and $k\leq n-k$ is MDP and each MDP convolutional code with $k\mid\delta$ is non-catastrophic; see \cite{AL:leftprime}.
\end{proof}

As mentioned before, for $k=n-k$, complete $j$-MDP convolutional codes defined with the parity-check matrix have the same erasure correcting properties as MDP convolutional codes defined via the generator matrix. Next, we will show that in this case the two criteria for a complete $j$-MDP convolutional code, coming from generator matrix and parity-check matrix, are equivalent.

\begin{lemma}\label{Lemma:GjcompHj}
If $n-k=k$ and $\nu=\mu$, then the indices $1\leq \ell_1< \hdots< \ell_{k(j+1+2\mu)}\leq n(j+1+\mu)$ fullfil conditions (i) and (ii) of Theorem \ref{Prop:NotTriviallyZeroMinorsG} if and only if the indices of the set $\{i_1,\hdots,i_{(n-k)(j+1)}\}=\{1,\hdots, n(j+1+\mu)\}\setminus\{\ell_1, \hdots, \ell_{k(j+1+2\mu)}\}$
fullfil conditions (i) and (ii) of Lemma \ref{lemma:nottriviallyzeroH}. 
\end{lemma}

\begin{proof}
Let $\ell_1,\dots,\ell_{k(j+1+2\mu)}$    be such that the conditions (i) and (ii) of Theorem \ref{Prop:NotTriviallyZeroMinorsG} are fulfilled
    and suppose that $i_1,\dots,i_{(j+1)(n-k)}$ do not fullfil the conditions of Lemma \ref{lemma:nottriviallyzeroH}. Then, for some $s\in\{1,\dots,j\}$ we have either or both of the following conditions:
    \begin{enumerate}
        \item $i_{(n-k)s+1}\leq sn$,
        \item $i_{(n-k)s}>(s+\nu)n$.
    \end{enumerate}

    On one hand, if for some $s$ we have that $i_{(n-k)s+1}\leq sn$, then $i_{ks+1}\leq sn$ since $n=2k$. This implies that at most $sn-(ks+1)=ks-1$ indices from $\{\ell_1, \hdots, \ell_{k(j+1+2\mu)}\}$ can be contained in $\{1,\hdots,sn\}$, i.e.
    $\ell_{ks}>sn$, which is a contradiction to condition (i) of Theorem \ref{Prop:NotTriviallyZeroMinorsG}.

    On the other hand, if $i_{(n-k)s}>(s+\mu)n$, we obtain $i_{ks}>(s+\mu)n$. This implies that at least $(s+\mu)n-(ks-1)=(s+2\mu)k+1$ indices from the set $\{\ell_1, \hdots, \ell_{k(j+1+2\mu)}\}$     have to be contained in $\{1,\hdots,(s+\mu)n\}$, i.e. $\ell_{(2\mu+s)k+1}\leq (s+\mu)n$ for some $s\in\{1,\dots,j\}$, which is equivalent to $\ell_{(\mu+\tilde{s})k+1}\leq \tilde{s}n$ for $\tilde{s}=s+\mu\in\{\mu+1,\hdots,\mu+j\}$ contradicting condition (ii) of Theorem \ref{Prop:NotTriviallyZeroMinorsG}.
        
Assume now that $i_1,\dots,i_{(j+1)(n-k)}$ are such that the conditions (i) and (ii) of Lemma \ref{lemma:nottriviallyzeroH} are fulfilled and suppose that    $\ell_1,\dots,\ell_{k(j+1+2\mu)}$ do not fullfil the conditions of Theorem \ref{Prop:NotTriviallyZeroMinorsG}. Then, for some $s\in\{1,\dots,j+\mu\}$ we have either or both of the following conditions:
    \begin{enumerate}
        \item $\ell_{sk}> sn$,
        \item $\ell_{(\mu+s)k+1}\leq sn$.
    \end{enumerate}

    If for some $s$ we have that $\ell_{sk}> sn$, then
    at least $sn-ks+1=ks+1$ indices from $i_1,\dots,i_{(j+1)(n-k)}$ have to be contained in $\{1,\hdots,sn\}$, i.e.
    $i_{ks+1}\leq sn$, which is a contradiction to condition (i) of Lemma \ref{lemma:nottriviallyzeroH}.
  
If $\ell_{(\mu+s)k+1}\leq sn$,
     at most $sn-((\mu+s)k+1)=sk-\mu k-1$ indices from the set $i_1,\dots,i_{(j+1)(n-k)}$
     can be contained in $\{1,\hdots,sn\}$, i.e. $i_{sk-\mu k}>sn$. Clearly, this is only possible if $s> \mu$ and we can set $\tilde{s}:=s-\mu\in\{1,\hdots,j\}$. Hence, we obtain $i_{k\tilde{s}}>(\tilde{s}+\mu)n$, a contradiction to condition (ii) of Lemma \ref{lemma:nottriviallyzeroH}.
\end{proof}

\begin{lemma}\cite[Lemma 8]{complementaryMinors}\label{Lem:ComplementaryMinors}
    Let $A\in\mathbb{F}^{m\times (m+\ell)}$ and $B\in\mathbb{F}^{\ell\times (\ell+m)}$ with $\ell, m\in\mathbb N$ be full row rank matrices such that $AB^\top={\bf 0}$. Given a full-size minor of $A$ constituted by the columns $i_1,i_2,\dots,i_m$, let us define the complementary full size minor of $B$ as the minor constituted by the complementary columns, i.e., by the columns $\{1,2,\dots,m+\ell\}\setminus\{i_1,i_2,\dots,i_m\}$. Then, the full size minors of $A$ are equal to the complementary full size minors of $B$ up to multiplication by a nonzero constant.
\end{lemma}

\begin{theorem}\label{th:2keqn}
    Let $n=2k$, $k\mid\delta$ and let $\mathcal{C}$ be a non-catastrophic convolutional code with row reduced (left prime) generator matrix $G(z)=\sum_{i=0}^{\mu}G_iz^i\in \mathbb{F}_q[z]^{k\times n}$ with $G_\mu$ full rank and with left prime parity-check matrix $H(z)=\sum_{i=0}^{\nu}H_iz^i\in \mathbb{F}_q[z]^{(n-k)\times n}$ where $H_{\nu}$ has full rank. Then, $\nu=\mu$ and all non-trivial full-size minors of 
$$\mathcal{G}^c_{j+\mu}=\begin{pmatrix}
        G_{\mu} & & \\
                \vdots & \ddots & \\
                       G_0 &  & G_\mu \\
        & \ddots   & \vdots  \\
        &  & G_0       
    \end{pmatrix}\in\mathbb F^{k(j+1+2\mu)\times n(j+1+\mu)}$$
are nonzero if and only if all non-trivial full-size minors of $$
    \mathcal{H}_{j}=\left(\begin{array}{cccccc}
     H_\nu   & \dots &   H_0 &      &        &\\
            &  H_\nu   & \dots &   H_0 &        &\\
            &       & \ddots &       & \ddots &\\
            &       &        &  H_\nu   & \dots &   H_0\\
    \end{array}\right)\in\mathbb F^{(j+1)(n-k)\times (j+1+\nu)n}
    $$
are nonzero.
\end{theorem}

\begin{proof}

As $G_{\mu}$ is full rank, all $k$ row degrees of $G(z)$ are equal to $\mu$ and $\delta=k\mu$.
Since $H(z)$ is left prime and $H_{\nu}$ has full rank, we obtain $\delta=(n-k)\nu$, we can conclude that $\nu=\mu$ because $n=2k$.

Due to the fact that $H(z)G(z)^\top =0$, one has
    $$
   \left(\begin{array}{ccccc}
     H_\nu   & \dots &   H_0 &      &        \\
                          & \ddots &       & \ddots &\\
                  &        &  H_\nu   & \dots &   H_0
    \end{array}\right)\cdot 
     \left(\begin{array}{ccccc}
     G_\mu^\top   & \dots &   G_0^\top &      &        \\
                          & \ddots &       & \ddots &\\
                  &        &  G_\mu^\top   & \dots &   G_0^\top 
    \end{array}\right)=0
    $$
Moreover, one has $(n-k)(j+1)+k(j+1+2\mu)=n(j+1)+2k\mu=n(j+1+\mu)$, i.e. the number of rows of $\mathcal{H}_j$ plus the number of columns of $(\mathcal{G}^{c}_{j+\mu})^\top$ is equal to the number of columns of $\mathcal{H}_j$ which equals the number of rows of $(\mathcal{G}^{c}_{j+\mu})^\top$. 

Hence, by Lemma \ref{Lem:ComplementaryMinors} a full-size minor of $\mathcal{H}_j$ is nonzero if and only if the complementary full-size minor of $\mathcal G^c_{j+\mu}$ is nonzero and by Lemma \ref{Lemma:GjcompHj} we have that the non-trivial full-size minors of $\mathcal{G}^c_{j+\mu}$ are the complementary minors of the non-trivial full-size minor of $\mathcal{H}_j$.
\end{proof}

In the following we present a general construction for complete MDP convolutional codes via the generator matrix. This shows that these codes exist for all $(n,k,\delta)$ with $k\mid\delta$ if the size of the finite field is sufficiently large.
The construction uses a lemma from \cite{dr16} and is similar to the one presented in \cite{J:CompleteMDP}, where the existence of complete MDP convolutional with $(n-k)\mid\delta$ was shown.

\begin{definition}\cite{dr16}
    Let $S_n$ be the symmetric group of order $n$. The determinant of an $n\times n$ matrix $A=[a_{i,j}]$ is given by $\operatorname{det}(A)=\sum_{\sigma\in S_n}(-1)^{\operatorname{sgn}(\sigma)}a_{1,\sigma(1)}\cdots a_{n,\sigma(n)}$. We call a product of the form $a_{1,\sigma(1)}\cdots a_{n,\sigma(n)}$ with $\sigma\in S_n$ a trivial term of the determinant if at least one component $a_{i,\sigma(i)}$ is equal to zero.
\end{definition}



\begin{lemma}\cite[Theorem 3.3]{dr16}\ \\
Let $\alpha$ be a primitive element of a finite field $\mathbb F=\mathbb F_{p^N}$ and $B=[b_{i,l}]$ be a matrix over $\mathbb F$ with the following properties
\begin{enumerate}
\item if $b_{i,l}\neq 0$, then $b_{i,l}=\alpha^{\beta_{i,l}}$ for a positive integer $\beta_{i,l}$
\item if $b_{i,l}=0$, then $b_{i',l}=0$ for any $i'>i$ or $b_{i,l'}=0$ for any $l'<l$
\item if $l<l'$, $b_{i,l}\neq 0$ and $b_{i,l'}\neq 0$, then $2\beta_{i,l}\leq\beta_{i,l'}$
\item if $i<i'$, $b_{i,l}\neq 0$ and $b_{i',l}\neq 0$, then $2\beta_{i,l}\leq\beta_{i',l}$.
\end{enumerate}
Suppose that $N$ is greater than any exponent of $\alpha$ appearing as a nontrivial term of any minor of $B$. Then $B$ has the property that each of its minors which is non-trivial is nonzero.
\end{lemma}

\begin{theorem}
    Let $n,k,\delta\in\mathbb{N}$ with $k<n$ and $k\mid\delta$ and let $\alpha$ be a primitive element of a finite field $\mathbb{F}=\mathbb{F}_{p^N}$ with $N>k(L+1+2\mu)\cdot 2^{(\mu+1)n+k-2}$. Then $G(z)=\sum_{i=0}^{\mu}G_iz^i$ with
    \begin{equation*}
        G_i=\left[\begin{array}{ccc}
            \alpha^{2^{in}} & \cdots & \alpha^{2^{(i+1)n-1}}  \\
             \vdots &             & \vdots\\
             \alpha^{2^{in+k-1}} & \cdots & \alpha^{2^{(i+1)n+k-2}}
        \end{array}\right]
    \end{equation*}
    for $i=0,1,\dots,\mu=\frac{\delta}{k}$ is the generator matrix of an $(n,k,\delta)$ complete MDP convolutional code.
\end{theorem}

\begin{proof}
  We have to show that each non-trivial full-size minor of $\mathcal{G}^c_{\mu+j}$ is nonzero. Permutation (reverse ordering) of the blocks of columns of $\mathcal{G}^c_{\mu+L}$, leads to the matrix $\left[\begin{array}{ccccc}
0 &  & G_0  \\ 
 & \text{\reflectbox{$\ddots$}} & \vdots  \\ 
G_0 &  & G_{\mu} \\
\vdots & \text{\reflectbox{$\ddots$}}\\
G_{\mu} & & 0
\end{array}\right]\in\mathbb F^{k(L+1+2\mu)\times n(L+1+\mu)}$, which fulfills the  conditions of the preceding lemma if
$N$ is larger than any exponent of $\alpha$ appearing as a nontrivial term of any minor of this matrix.
The largest possible value for such an exponent is upper bounded by $k(L+1+2\mu)\cdot 2^{(\mu+1)n+k-2}$. 
\end{proof}

\begin{example}
    In this example we consider a $(3,1,1)$ code with generator matrix $G(z)=\sum_{i=0}^{\mu}G_iz^i$ defined by 
    $$
    G_0=\left[\begin{array}{ccc} \alpha & \alpha^2 & \alpha^4 \end{array}\right];\quad G_1=\left[\begin{array}{ccc} \alpha^8 & \alpha^{16} & \alpha^{32} \end{array}\right].
    $$
    One gets a complete MDP code with $k\mid\delta$ and $(n-k)\nmid\delta$ over the field $\mathbb{F}_{p^N}$ if $N>6\cdot 2^5$.
\end{example}

Note that the extension we propose in the definition of complete $j$-MDP convolutional codes enlarges this family of codes. The previous construction shows that new kinds of these codes can be obtained in a very similar way as with the previous definition.
The construction of (complete) MDP convolutional codes over small finite fields is a challenging problem, no matter whether they are defined via the parity-check matrix or via the generator matrix.

\section{Performance Analysis}\label{sec_complexity}

In this section, we calculate the complexity, i.e. the number of necessary operations in the finite field $\mathbb F_q$, of our erasure decoding procedure using the generator matrix of the convolutional code and compare it
with the complexity of the decoding method introduced in \cite{TRS:Decod}, which uses the parity-check matrix of the convolutional code. The complexity of both mentioned ways of erasure decoding is determined by the complexity of solving linear systems of equations. Therefore, we use that according to \cite{complexity} the complexity of solving a linear system with $m$ equations and $x$ unknowns is $O(m^{0.8}\cdot x^2)$.

For the forward decoding method 
presented in the Section \ref{sec_decoding} using the generator matrix, it is needed to solve
the equation system (\ref{smallcomplete}), where the unknowns are the symbols of the messages.
Hence, in each window of size $(j+1)n$ containing $e$ erasures we have to solve an equation system with $(j+1)k$ unknowns and $(j+1)n-e$ equations, which results in a decoding complexity of $O(((j+1)n-e)^{0.8}((j+1)k)^2)$. When using forward decoding with the parity-check matrix, the complexity for this is $O(((j+1)(n-k))^{0,8}e^2)$ and we see that using the generator matrix has a better performance when the amount of erasures is large. Using MDP convolutional codes, we can correct up to $e=(n-k)(j+1)$ erasures in such a window and obtain complexities of $O(((j+1)k)^{2.8})$ using the generator matrix and $O(((j+1)(n-k))^{2.8})$ using the parity-check matrix. Hence, decoding with the generator matrix is more efficient if and only if $k<n-k$, which is the same parameter range for which it also has a better guard space recovering rate; see Table \ref{tab:RcoveringRatio}.

In summary, without the need for guard spaces, the same amount of erasures can be corrected with generator matrix and parity-check matrix and what is more efficient depends on whether $k$ or $n-k$ is smaller.

For the calculation of guard spaces the computational complexities are in both cases $O(((n(\mu+j+1))^{0.8}(k(2\mu+j+1))^2)$. 
Nonetheless, if $k$ is smaller than $n-k$, we can tolerate more erasures when using the generator matrix, while when $n-k$ is smaller, we can tolerate more erasures with the parity-check matrix. 
However, there are situations when $n-k<k$ and it is still better to use the generator matrix instead of the parity-check matrix.
First of all, this is obviously true if the code is catastrophic and we do not have a parity-check matrix. Secondly, note that using the generator matrix, we immediately recover the messages $u_i$, while using the parity-check matrix, we only recover $v(z)$. If $v(z)$ can be recovered, it is easy to calculate all coefficients of $u(z)$ from the coefficients of $v(z)$. However, if at some point too many erasures occurred such that some symbols have to be declared as lost when using the parity-check matrix to compute a new guard space, we recover $\bfv_t,\dots\bfv_{t+j}$, but we still have to get the corresponding messages $\bfu_{t-\mu},\dots,\bfu_{t+j}$. To do so, one needs to solve the system 
\begin{equation}\label{eq:SolveG}
(\bfu_{t-\mu},\dots,\bfu_{t+j})\mathcal{G}^c_{j}=(\bfv_t,\dots\bfv_{t+j}),   
\end{equation}
where the unknowns are the messages we want to recover. The complexity for this is $\mathcal{O}(((j+\mu+1)k)^2((j+1)n)^{0.8})$, which is neglectable with respect to the recovery of the erasures.

\section{Conclusions}
We provided an erasure decoding algorithm for arbitrary not necessarily non-catastrophic convolutional codes using only the generator matrix of the code. To obtain codes very well suited for this algorithm, we define a new class of complete MDP convolutional codes. Evaluating the performance of the new algorithm and codes we see that for $k<n-k$ they perform better than previous algorithms using the parity-check matrix. A future line of work is to investigate how the ideas of this paper can be generalized to convolutional codes over finite rings.

\section*{Acknowledgements}
The first author was supported by the German research foundation, project number 513811367. The second author was supported by CIDMA under the FCT
(Portuguese Foundation for Science and Technology) 
Multi-Annual Financing Program for R$\&$D Units. The third author is supported by the Grundlagenforschungsfond (GFF) of the University of St.Gallen, project no.\ 2260780.

\bibliographystyle{abbrv} 
\bibliography{arxivbib}
\end{document}